\documentclass[5p,twocolumn]{elsarticle}

\usepackage{hyperref}
\usepackage{amsmath}
\usepackage{amssymb}
\usepackage{enumerate}
\usepackage[all]{xy}
\usepackage{natbib}

\journal{Automatica}

\biboptions{authoryear}

\def \R{{\mathbb R}}
\def \N{{\mathbb N}}

\def \KL{\mathcal{KL}}

\def \ki{\mathcal{K}_{\infty}}

\def \U{\mathcal{U}}
\def \S{\mathcal{S}}
\def \T{\mathcal{T}}
\def \comp{\,{\scriptstyle\circ}\,}

\def\K{\mathcal{K}}
\def\Ki{\K_{\infty}}
 \def\mer{\hfill $\circ$}

\DeclareMathOperator*{\esssup}{ess.sup}

\newtheorem{teo}{Theorem}[section]
\newtheorem{lema}[teo]{Lemma}
\newtheorem{cor}[teo]{Corollary}
\newtheorem{prop}[teo]{Proposition}
\newtheorem{claim}{Claim}
\newdefinition{defin}[teo]{Definition}
\newdefinition{as}{Assumption}
\newdefinition{rem}{Remark}
\newproof{proof}{\textbf{Proof}}

\renewcommand{\qed}{$\hfill\blacksquare$}

\begin{document}

\begin{frontmatter}

\title{Strong ISS implies Strong iISS for Time-Varying Impulsive Systems}


\author[HHaddress]{Hernan Haimovich\corref{mycorrespondingauthor}}
\cortext[mycorrespondingauthor]{Corresponding author}
\ead{haimovich@cifasis-conicet.gov.ar}

\author[JLaddress]{Jos\'e L. Mancilla-Aguilar}
\ead{jmancill@itba.edu.ar}

\address[HHaddress]{International Center for Information and Systems Science (CIFASIS), CONICET-UNR, Ocampo y Esmeralda, 2000 Rosario, Argentina.}
\address[JLaddress]{Departamento de Matem\'atica, Instituto Tecnol\'ogico de Buenos Aires, Av. Eduardo Madero 399, Buenos Aires, Argentina.}

\begin{abstract}
  For time-invariant (nonimpulsive) systems, it is already well-known that the input-to-state stability (ISS) property is strictly stronger than integral input-to-state stability (iISS). Very recently, we have shown that under suitable uniform boundedness and continuity assumptions on the function defining system dynamics, ISS implies iISS also for time-varying systems. In this paper, we show that this implication remains true for impulsive systems, provided that asymptotic stability is understood in a sense stronger than usual for impulsive systems.
\end{abstract}


\end{frontmatter}


\section{Introduction}
\label{sec:introduction}
One of the main issues in control system theory concerns understanding the dependence of state trajectories on inputs. In this regard, the input-to-state stability (ISS) and integral-ISS (iISS) are arguably the most important and useful state-space based nonlinear definitions of stability for systems with inputs. 

The notions of ISS and iISS, originally introduced for time-invariant continuous-time systems in \citet{sontag_tac89} and \citet{sontag_scl98},  respectively, were subsequently extended and studied for other classes of systems: time-varying systems \citep{edwlin_cdc00}, discrete-time systems \citep{jiawan_auto01}, switched systems \citep{mangar_scl01, haiman_tac18}, impulsive systems \citep{heslib_auto08}, hybrid systems \citep{caitee_scl09,norkha_mcss17} and infinite dimensional systems \citep{dasmir_mcss13,mirwir_tac17}.

A natural question regards the exact relationship between the ISS and iISS properties. Since the introduction of the iISS property it is known that ISS implies iISS and that the converse does not hold for time-invariant continuous-time systems \citep{sontag_scl98}. The same implication was proved for discrete-time systems \citep{ang_scl99}, switched systems under arbitrary switching \citep{mangar_scl01} and hybrid systems \citep{norkha_mcss17}, assuming time-invariance. The corresponding proofs employ Lyapunov characterizations of the ISS or of the global uniform asymptotic stability (GUAS) properties in a fundamental way. This hinders the extension to classes of systems for which Lyapunov characterizations do not exist, such as switched systems under restricted switching or impulsive systems. Very recently, \citet{haiman_auto19} proved that ISS implies iISS for families of time-varying and switched nonlinear systems without resorting to any Lyapunov converse theorem, and, in this way, opening the door to proving the implication for other types of systems.

This paper deals with impulsive systems with inputs, i.e. dynamical systems whose state evolves continuously most of the time but may exhibit jumps (discontinuities) at isolated time instants, and where the inputs affect both the flow (i.e. the continuous evolution) and the jump equations \citep{yanpen_mcs19}. Sufficient conditions for ISS and iISS of impulsive systems with inputs, based on Lyapunov-type functions, have been derived in \citet{heslib_auto08}. Since the appearance of \citet{heslib_auto08}, many works have addressed the stability of impulsive systems with inputs from ISS-related standpoints, giving sufficient conditions for the ISS and/or iISS in terms of Lyapunov functions \citep{chezhe_auto09,liuliu_auto11,daskos_nahs12,dasmir_siamjco13,liudou_scl14,lizha_auto17,dasfek_nahs17,lili_scl18,penden_scl18,peng_ietcta18,ninhe_is18,lili_mcs19,manhai_tac19arxiv}. In addition, some results for hybrid systems may also be applicable to impulsive systems  \citep{libnes_tac14,miryan_auto18,liuhil_amc18}.

Despite the great progress of the stability theory of impulsive systems with inputs during the last decade, up to our best knowledge the exact relationship between the ISS and iISS properties has not yet been established for this type of systems. The main contribution of the current paper is thus solving this open problem by proving that ISS implies iISS. The implication is proved assuming that the ISS and iISS properties are understood in a stronger sense than is usually considered in the literature of impulsive systems, more akin to that employed for hybrid systems. As is well-known, the ISS/iISS properties impose a bound on the state trajectory comprising a decaying-to-zero term whose amplitude depends on the initial state value, and an input magnitude/energy-dependent term. As already explained in \citet{heslib_auto08}, the decaying term in the ISS/iISS definitions employed for impulsive systems decays as time elapses but is insensitive to the occurrence of jumps. In this paper, we consider definitions of ISS/iISS where the decaying term decreases also when a jump occurs (see Definition \ref{def:stab} below), in agreement with those considered in the context of hybrid systems \citep{caitee_cdc05,caitee_scl09,norkha_mcss17}. As a corollary of our main result, we obtain that ISS implies iISS in the usual sense when the impulse-time sequence satisfies a specific bound on the number of impulse times on each bounded interval. This condition is satisfied, for example, when the impulse-time sequence is such that the flow periods (i.e. between jumps) have a minimum or average dwell time.

The current paper generalizes some of our previous results \citep{haiman_tac18,haiman_aadeca18,haiman_auto19,haiman_rpic19}. Our proof strategy conceptually follows that of \citet{haiman_auto19}, in the sense of being based on bounding the difference between state trajectories. The current results cannot be obtained directly (mutatis mutandis) from the previous ones, mainly because we do not require the jump maps to satisfy any kind of Lipschitz continuity property. This led to the development of novel techniques for comparing trajectories, especially suited to impulsive systems without Lipschitz continuity of the jump maps. The specific similarities and differences with respect to our previous work are explained as appropriate along the text.


The remainder of the paper is organized as follows. This section ends with a brief description of the notation employed. In Section~\ref{sec:stab-impuls-syst}, we precisely explain the type of systems considered and the stability concepts employed. In Section~\ref{sec:iISS-char}, we provide a characterization of the strong iISS property. This characterization is employed in Section~\ref{sec:iss-iiss} in order to establish that strong ISS implies strong iISS. The proofs of some technical intermediate results are given in Section~\ref{sec:proofs}. Conclusions are given in Section~\ref{sec:conclusions}.

\textbf{Notation.} $\N$, $\R$, $\R_{>0}$ and $\R_{\ge 0}$ denote the natural numbers, reals, positive reals and nonnegative reals, respectively. $|x|$ denotes the Euclidean norm of $x \in \R^p$. We write $\alpha\in\K$ if $\alpha:\R_{\ge 0} \to \R_{\ge 0}$ is continuous, strictly increasing and $\alpha(0)=0$, and $\alpha\in\Ki$ if, in addition, $\alpha$ is unbounded. We write $\beta\in\KL$ if $\beta:\R_{\ge 0}\times \R_{\ge 0}\to \R_{\ge 0}$, $\beta(\cdot,t)\in\Ki$ for any $t\ge 0$ and, for any fixed $r\ge 0$, $\beta(r,t)$ monotonically decreases to zero as $t\to \infty$. From any function $h:I\subset \R\to \R^p$, $h(t^-)$ and $h(t^+)$ denote, respectively, the left and right limits of $h$ at $t\in \R$, when they exist and are finite. For every $n\in\N$ and $r\ge 0$, we define the closed ball $B_r^n := \{x\in\R^n : |x| \le r\}$. Without risk of confusion, if $\gamma = \{\tau_k\}_{k=1}^N$, then $\gamma$ can be interpreted as both the sequence $\{\tau_k\}_{k=1}^N$ and the set $\{\tau_k:k\in\N,\ k\le N\}$ (even if $N=\infty$). For $a,b\in\R$, we define $a \wedge b := \min\{a,b\}$.

\section{Stability of Impulsive Systems with Inputs}
\label{sec:stab-impuls-syst}

\subsection{Impulsive systems with inputs}
\label{sec:prel-defs}

Consider the time-varying impulsive system with inputs $\Sigma$ defined by the equations
 \begin{subequations}
   \label{eq:is}
   \begin{align}
     \label{eq:is-ct}
     \dot{x}(t) &=f(t,x(t),u(t)),\phantom{x(t^-)+g^-}\quad\text{for } t\notin \gamma,    \displaybreak[0] \\
     \label{eq:is-st}
     x(t) &=x(t^-)+g(t,x(t^-),u(t)),\phantom{f} \quad\text{for } t\in \gamma,
   \end{align}
 \end{subequations}
where $t\ge 0$, the state variable $x(t)\in \R^n$, the input variable $u(t)\in \R^m$ and $f$ and $g$ are functions from $\R_{\ge 0}\times \R^n\times \R^m$ to $\R^n$ such that $f(t,0,0)=0$ and $g(t,0,0)=0$ for all $t\ge 0$, and the impulse-time sequence $\gamma=\{\tau_k\}_{k=1}^{N} \subset (0,\infty)$, with $N$ finite or $N=\infty$. We shall refer to $f$ and to (\ref{eq:is-ct}) as, respectively, the flow map and the flow equation and to $g$ and to (\ref{eq:is-st}) as, respectively, the jump map and the jump equation. By ``input'', we mean a Lebesgue measurable and locally essentially bounded function $u:[0,\infty)\to \R^m$; we denote by $\U$ the set of all the inputs. As is usual for impulsive systems, we only consider impulse-time sequences $\gamma=\{\tau_k\}_{k=1}^N$ that are strictly increasing and have no finite limit points, i.e. $\lim_{k\to \infty}\tau_k=\infty$ when the sequence is infinite; we employ $\Gamma$ to denote the set of all such impulse-time sequences. For any sequence $\gamma = \{ \tau_{k} \}^N_{k=1} \in \Gamma$ we define for convenience $\tau_0=0$; 
 nevertheless, $\tau_0$ is never an impulse time, because $\gamma \subset (0,\infty)$ by definition.

In order to guarantee the existence of Carath\'eodory solutions of the differential equation $\dot{x}(t)=f(t,x(t),u(t))$, we assume that $f(t,\xi,\mu)$ is Lebesgue measurable in $t$, continuous in $(\xi,\mu)$ and that for every compact interval $I\subset \R_{\ge 0}$ and every compact set $K\subset \R^n\times \R^m$ there exists an integrable function $m:I\to \R$ such that $|f(t,\xi,\mu)|\le m(t)$ for all $(t,\xi,\mu)\in I\times K$. Under these conditions, for each input $u\in \U$ the map $f_u(t,\xi)=f(t,\xi,u(t))$ satisfies the standard Carath\'eodory conditions \citep[see][]{hale_book80} and hence the (local) existence of solutions of the differential equation $\dot x(t)=f(t,x(t),u(t))$ is ensured.

The impulsive system $\Sigma$ is completely determined by the sequence of impulse times $\gamma$ and the flow and jump maps $f$ and $g$. Hence, we write $\Sigma=(\gamma,f,g)$.  Given $\gamma \in \Gamma$ and an interval $I\subset [0,\infty)$, we define $n^\gamma_I$ as the number of elements of $\gamma$ that lie in the interval $I$: 
\begin{align}
  n^\gamma_I &:= \# \big[ \gamma\cap I \big].
\end{align}

A solution of $\Sigma=(\gamma,f,g)$ corresponding to an initial time $t_0\ge 0$, an initial state $x_0\in \R^n$ and an input $u\in \U$ is a 
function $x:[t_0,T_x)\to \R^n$ such that: 
\begin{enumerate}[i)]
\item $x(t_0)=x_0$; 
\item $x$ is locally absolutely continuous on each interval $J = [t_1,t_2) \subset [t_0,T_x)$ without points of $\gamma$ in its interior, and $\dot{x}(t)=f(t,x(t),u(t))$ for almost all $t\in J$; and \label{item:solflow}
\item for all $t\in \gamma \cap (t_0,T_x)$, the left limit $x(t^-)$ exists and is finite, and it happens that $x(t) = x(t^-)+g(t,x(t^-),u(t))$.\label{item:soljump}
\end{enumerate}
Note that \ref{item:solflow}) implies that for all $t\in [t_0,T_x)$, $x(t)=x(t^+)$, i.e. $x$ is right-continuous at $t$.

The solution $x$ is said to be maximally defined if no other solution $y:[t_0,T_y)$ satisfies $y(t) = x(t)$ for all $t\in [t_0,T_x)$ and has $T_y > T_x$. 
 We will use $\T_{\Sigma}(t_0,x_0,u)$ to denote the set of maximally defined solutions of $\Sigma$ corresponding to initial time $t_0$, initial state $x_0$ and input $u$. 
%
Every solution $x\in\T_{\Sigma}(t_0,x_0,u)$ with $t_0 \ge 0$, $x_0\in\R^n$ and $u\in\U$ satisfies
\begin{multline}
  \label{eq:solintform}
  x(t) = x(t_0) + \int_{t_0}^t f(s,x(s),u(s)) ds\\ +
  \sum_{\tau\in\gamma\cap(t_0,t]} g(\tau,x(\tau^-),u(\tau)),\quad\forall t\in [t_0,T_x).
\end{multline}
\begin{rem}
  Note that even if $t_0 \in \gamma$, any solution $x\in\T_{\Sigma}(t_0,x_0,u)$ begins its evolution by ``flowing'' and not by ``jumping''. This is because in item~\ref{item:soljump}) above, the time instants where jumps occur are those in $\gamma \cap (t_0,T_x)$.\mer
\end{rem}

\subsection{Families of impulsive systems}
\label{sec:families}

Often one is interested in determining whether some stability property holds not just for a single impulse-time sequence $\gamma\in\Gamma$ but also for some family $\S \subset \Gamma$. For example, the family $\S$ could contain all those impulse-time sequences having some minimum, maximum or average dwell time. Another situation of interest is to determine if some stability property holds not just for a single pair of functions $(f,g)$ but also for all pairs $(f,g)$ 
belonging to some given set $\mathcal{F}$.  
To take into account these and other situations, we consider a parametrized family $\Sigma_\Lambda := \{\Sigma_{\lambda}=(\gamma_{\lambda},f_{\lambda},g_{\lambda})\}_{\lambda \in \Lambda}$ of impulsive systems with inputs, where $\Lambda$ is an index set (i.e. an arbitrary nonempty set). For example, if we are interested in studying stability properties of systems modelled by (\ref{eq:is}) which hold uniformly over a class $\S\subset \Gamma$, then we set $\S$ as the index set, and consider the parametrized family of systems $\{\Sigma_{\gamma}=(\gamma,f,g)\}_{\gamma \in \S}$. By taking as index set $\Lambda = \mathcal{F}$ and considering the family $\{\Sigma_{(f,g)} = (\gamma,f,g) \}_{(f,g) \in \Lambda}$ we can handle the other mentioned situation. Another interesting situation we can handle in this way is that of impulsive switched systems \citep[see][for details]{manhai_tac19arxiv}.

\subsection{Stability definitions}
\label{sec:stab-defs}

In the context of impulsive systems, the input can be interpreted as having both a continuous-time and an impulsive component. From (\ref{eq:is-st}) one observes that the values of $u$ at the instants $t\in\gamma$ may instantaneously affect the state trajectory. For this reason, input bounds suitable for the required stability properties have to account for the instantaneous values $u(t)$ at $t\in\gamma$. Given an input $u \in \U$, an impulse-time sequence $\gamma \in \Gamma$, an interval $I\subset \R_{\ge 0}$, and functions $\rho_1,\rho_2\in\Ki$, we thus define
\begin{align}
  \| u_I \|_{\infty,\gamma} &:= \max\left\{ \esssup_{t\in I} |u(t)| , \sup_{t\in \gamma\cap I} |u(t)| \right\}, \label{eq:iss-norm}\\
  \| u_I \|_{\rho_1,\rho_2,\gamma} &:= \int_I \rho_1(|u(s)|) ds + \sum_{s\in \gamma\cap I} \rho_2(|u(s)|). \label{eq:iiss-norm}
\end{align}
When $I=[0,\infty)$ we simply write $u$ instead of $u_I$. These definitions are in agreement with those employed in \citet{caitee_scl09, norkha_mcss17} in the context of hybrid systems. In what follows, $\mathbf{0}$ denotes the identically zero input.
\begin{defin}
  \label{def:stab}
  We say that the parametrized family $\Sigma_\Lambda = \{\Sigma_{\lambda} = (\gamma_{\lambda},f_{\lambda},g_{\lambda})\}_{\lambda \in \Lambda}$ of impulsive systems is  
  \begin{enumerate}[a)]
  \item strongly 0-GUAS if there exist $\beta \in \KL$ such that for all $\lambda\in \Lambda$, $t_0\ge 0$, $x_0\in \R^n$, and $x\in \T_{\Sigma_{\lambda}}(t_0,x_0,\mathbf{0})$, it happens that for all $t\in [t_0,T_x)$,
    \begin{align}
      \label{eq:0-guas}
      |x(t)| &\le \beta \left (|x_0|,t-t_0+n^{\gamma_{\lambda}}_{(t_0,t]} \right ).
    \end{align}
  \item strongly ISS if there exist $\beta \in \KL$ and $\rho \in \ki$ such that 
    \begin{align}
      \label{eq:ciss}
      \hspace{-2mm}|x(t)| &\le \beta \left (|x_0|,t-t_0+n^{\gamma_{\lambda}}_{(t_0,t]}\right )+\rho(\|u_{(t_0,t]}\|_{\infty,\gamma_{\lambda}});
    \end{align}
  \item strongly iISS if there exist $\beta \in \KL$ and $\alpha,\rho_1,\rho_2 \in \ki$ such that 
    \begin{align}
      \label{eq:ciiss}
      \hspace{-3mm}\alpha(|x(t)|) &\le \beta \left (|x_0|,t-t_0+n^{\gamma_{\lambda}}_{(t_0,t]} \right) + \| u_{(t_0,t]} \|_{\rho_1,\rho_2,\gamma_{\lambda}};
    \end{align}
  \item UBEBS \citep{angson_dc00} if there exist $\alpha,\rho_1,\rho_2\in\ki$ and $c\ge 0$ such that 
    \begin{align}
      \label{eq:cubebs}
      \alpha(|x(t)|) &\le |x_0| + \| u_{(t_0,t]} \|_{\rho_1,\rho_2,\gamma_{\lambda}} + c;
    \end{align}
  \end{enumerate}
  where (\ref{eq:ciss})--(\ref{eq:cubebs}) hold for all $\lambda\in \Lambda$, $t_0\ge 0$, $x_0\in \R^n$, $u\in \U$, $x\in \T_{\Sigma_{\lambda}}(t_0,x_0,u)$, and $t\in [t_0,T_x)$. The pair $(\rho_1,\rho_2)$ in (\ref{eq:ciiss}) or (\ref{eq:cubebs}) will be referred to as an iISS or UBEBS gain, respectively.
\end{defin}
\begin{rem}
\label{rem:caus-ISS}
Due to causality and the Markov property, equivalent definitions are obtained if $u_{(t_0,t]}$ is replaced by $u$ in (\ref{eq:ciss}), (\ref{eq:ciiss}) or (\ref{eq:cubebs}). Note that we do not require the solutions of (\ref{eq:is}) to be defined for all $t\ge t_0$ in the definitions of the different stability properties. Nevertheless, well-known results for ordinary differential equations ensure the existence of the solution on $[t_0,\infty)$ in each case. \mer
\end{rem}
\begin{rem}
  \label{rem:sISS_s0-GUAS}
  It is evident that strong ISS implies strong 0-GUAS (just set $u=\mathbf{0}$).
\end{rem}

All the properties in Definition~\ref{def:stab} are uniform with respect to both initial time $t_0$ and the different systems within the family $\Sigma_\Lambda$. The ISS and iISS properties are called ``strong'' because the decaying term given by the function $\beta$ forces additional decay whenever a jump occurs. The corresponding weak properties are obtained by replacing the second argument of $\beta$ by $t-t_0$ \citep[see][]{manhai_tac19arxiv}. Strong ISS (and iISS) is in agreement with the ISS property for hybrid systems as in \citet{libnes_tac14}. 

The strong and weak ISS/iISS become equivalent under the following condition, which is satisfied when the time periods between impulses have a minimum or average dwell time. 
\begin{defin}
  \label{def:UIB}
  Consider a set $\S\subset\Gamma$ of impulse-time sequences. We say that $\S$ is uniformly incrementally bounded (UIB) if there exists a continuous and nondecreasing function $\phi : \R_{\ge  0} \to \R_{\ge 0}$ so that $n^{\gamma}_{(t_0,t]} \le \phi(t-t_0)$ for every $\gamma\in \S$ and all $t>t_0\ge 0$. 
\end{defin}

The proof of the following result can be obtained following the lines of that of Proposition 2.3 in \citet{manhai_tac19arxiv}.
\begin{prop} 
  \label{prop:weakstrong}
  Let $\Sigma_\Lambda = \{\Sigma_{\lambda} = (\gamma_{\lambda}, f_{\lambda}, g_{\lambda})\}_{\lambda \in \Lambda}$. Suppose that $\{\gamma_{\lambda}:\lambda \in \Lambda\}$ is UIB. Then $\Sigma_\Lambda$ is strongly ISS (resp. iISS) if and only if it is weakly ISS (iISS).
\end{prop}

\section{A characterization of iISS}
\label{sec:iISS-char}
In this section we will show that under suitable hypotheses, the strong iISS of a parametrized family of impulsive systems with inputs is equivalent to the combination of UBEBS and strong 0-GUAS of the family. 

\subsection{Assumptions and statement}
\label{sec:ass-stat}

First, we note that if jumps do not occur ($\gamma = \emptyset$), then (\ref{eq:is}) becomes the type of system considered in \citet{haiman_tac18}. 
We thus require that the flow maps satisfy the conditions in Assumption~1 of \cite{haiman_tac18}. 
\begin{as}
  \label{as:ass-f-1}
  The functions $f_{\lambda} : \R_{\ge 0} \times \R^n \times \R^m \to \R^n$, $\lambda \in \Lambda$, satisfy the following:
  \begin{enumerate}[i)]
  \item \label{item:fbound} there exist $\nu_f \in \K$ and a nondecreasing function $N_f:\R_{\ge 0}\to \R_{>0}$ such that for all $\lambda\in \Lambda$, $|f_{\lambda}(t,\xi,\mu)|\le N_f(|\xi|)(1+\nu_f(|\mu|))$ for all $(t,\xi,\mu) \in \R_{\ge 0} \times \R^n \times \R^m$;
  \item for every $r>0$ and $\varepsilon >0$ there exists $\delta>0$ such that for all $\lambda\in \Lambda$ and all $t\ge 0$, $|f_{\lambda}(t,\xi,\mu) - f_{\lambda}(t,\xi,0)|<\varepsilon$ if $|\xi|\le r$ and $|\mu| \le \delta$;\label{item:fcont}
  \item  \label{item:lips} $f_{\lambda}(t,\xi,0)$ is locally Lipschitz in $\xi$, uniformly in $t$ and $\lambda$, i.e. for every $R>0$ there is a constant $L_R\ge 0$ so that for every $\lambda\in \Lambda$, $\xi_1,\xi_2 \in B_R^n$ and $t\ge 0$ it happens that $|f_{\lambda}(t,\xi_1,0) - f_{\lambda}(t,\xi_2,0)| \le L_{R} |\xi_1-\xi_2|$.
  \end{enumerate}
\end{as}
All the conditions imposed by Assumption~\ref{as:ass-f-1} on the flow maps are uniform over all the systems in the family. Item~\ref{item:fbound}) imposes a bound that is, in addition, uniform over all values of the time variable. Item~\ref{item:fcont}) requires a kind of continuity in the input variable at its zero value, uniformly over time and over states in compact sets. Item~\ref{item:lips}) requires that the flow map of the zero-input system be locally Lipschitz in the state variable, uniformly over time.

The Lipschitz condition in item~\ref{item:lips}) is required in order to allow the application of Gronwall inequality. If all the conditions of Assumption~\ref{as:ass-f-1} were imposed on the jump maps $g_\lambda$ as well, then the required characterization of strong iISS would follow, mutatis mutandis, from \citet{haiman_rpic19}. 
However, imposing such a Lipschitz continuity requirement on the jump maps is restrictive and unnecessary. We will thus require the following conditions.
\begin{as}
  \label{as:ass-g-1}
  The functions $g_{\lambda} : \R_{\ge 0} \times \R^n \times \R^m \to \R^n$, $\lambda \in \Lambda$, satisfy the following:
  \begin{enumerate}[i)]
  \item \label{item:gbound} there exist $\nu_g \in \K$ and a nondecreasing function $N_g:\R_{\ge 0}\to \R_{>0}$ such that for all $\lambda\in \Lambda$, $|g_{\lambda}(t,\xi,\mu)|\le N_g(|\xi|)(1+\nu_g(|\mu|))$ for all $(t,\xi,\mu) \in \R_{\ge 0} \times \R^n \times \R^m$;
  \item for every $r>0$ and $\varepsilon >0$ there exists $\delta>0$ such that for all $\lambda\in \Lambda$ and all $t\ge 0$, $|g_{\lambda}(t,\xi,\mu) - g_{\lambda}(t,\xi,0)|<\varepsilon$ if $|\xi|\le r$ and $|\mu| \le \delta$;\label{item:gcont}
  \item  \label{item:0-gcont} $g_{\lambda}(t,\xi,0)$ is continuous in $\xi$, uniformly in $t$ and $\lambda$, i.e. for every $R>0$ there is a function $\omega_{R}\in \ki$ so that for every $\xi_1,\xi_2 \in B_R^n$, $t\ge 0$ and $\lambda \in \Lambda$, it happens that $|g_{\lambda}(t,\xi_1,0) - g_{\lambda}(t,\xi_2,0)| \le \omega_{R}(|\xi_1-\xi_2|)$.
  \end{enumerate}
\end{as}
Items~\ref{item:gbound}) and~\ref{item:gcont}) of Assumption~\ref{as:ass-g-1} are identical to those of Assumption~\ref{as:ass-f-1}. By constrast, the Lipschitz continuity requirement of Assumption~\ref{as:ass-f-1}\ref{item:lips}) has been replaced by just continuity, keeping the corresponding uniformity with respect to the other variables. The removal of the Lipschitz continuity requirement on the jump maps causes the proof of our current results to become substantially different and harder than that of the previous ones \citep{haiman_aadeca18,haiman_auto19,haiman_rpic19}.

The main result of this section is the following charaterization of strong iISS for parametrized families of impulsive systems with inputs. 
\begin{teo}
  \label{UBEBSand0GASiffiiss}
  Consider the parametrized family $\Sigma_\Lambda = \{\Sigma_{\lambda}=(\gamma_{\lambda},f_{\lambda},g_{\lambda})\}_{\lambda\in \Lambda}$ and let Assumptions \ref{as:ass-f-1} and \ref{as:ass-g-1} hold. Then $\Sigma_\Lambda$ is strongly iISS if and only if it is strongly 0-GUAS and UBEBS.
\end{teo}
The proof of Theorem~\ref{UBEBSand0GASiffiiss} is given in Section~\ref{sec:proof-iiss-charact}. Note that Theorem~\ref{UBEBSand0GASiffiiss} does not require uniqueness of solutions under nonzero inputs because the local Lipschitz continuity of the flow maps imposed by Assumption~\ref{as:ass-f-1}\ref{item:lips}) applies only under zero input.

\subsection{Preliminary results}
\label{sec:prel-results-iISS}

The proof of Theorem \ref{UBEBSand0GASiffiiss} requires some preliminary lemmas. The first of these is a type of generalized Gronwall inequality for impulsive systems. The proof is given in Section~\ref{sec:proof-lem-ggi2}.
\begin{lema}
  \label{lem:ggi2}
  Let $0 \le t_0 < T$ and let $y:[t_0,T] \to \R_{\ge 0}$ be a right-continuous function having a finite number $N$ of points of discontinuity $s_1,\ldots,s_N$ satisfying $t_0<s_1<\ldots<s_N\le T$. Let $y$ be such that the left-limit $y(s_j^-)$ exists for all $j=1,\ldots,N$. Let $p\in\R_{\ge 0}$, let $a : \R_{\ge 0} \to \R_{\ge 0}$ be locally integrable, let $\{c_k\}_{k=1}^{\infty}$ be a sequence of nonnegative numbers, and let $\omega \in \ki$. Let $\sigma=\{s_k\}_{k=1}^N$ and define $c:\sigma\to \R_{\ge 0}$ via $c(s_j)=c_j$. If $y$ satisfies
  \begin{align}
    \label{eq:yGron1}
    y(t) &\le p + \int_{t_0}^t a(s) y(s) ds +  \sum_{s\in\sigma\cap (t_0,t]} c(s) \omega(y(s^-))
  \end{align}
  for all $t \in [t_0,T]$, then in the same time interval $y$ satisfies 
  \begin{align}
    y(t) &\le h_k^{t_0}(p,t),
  \end{align}
  where $k=n^{\sigma}_{(t_0,t]}$, and the functions $h_j^{t_0}:\R_{\ge 0}\times [t_0,\infty) \to \R_{\ge 0}$, $j=0,1,\ldots$, are recursively defined as follows
 \begin{align*}
   h_0^{t_0}(p,&t) =p e^{\int_{t_0}^t a(s) ds},\quad\text{and, for $j\ge 1$,}\\
   h_j^{t_0}(p,&t) = h_{j-1}^{t_0}(p,t)+\\
   &c_{j} e^{\int_{t_0}^t a(s) ds} \sup_{t_0 \le s \le t} \left[ \omega(h_{j-1}^{t_0}(p,s))e^{-\int_{t_0}^s a(\tau) d\tau} \right].
 \end{align*}
\end{lema}

The function $\omega$ on the right-hand side of inequality (\ref{eq:yGron1}) makes the third term therein not necessarily affine in $y$. This enables the application of Lemma~\ref{lem:ggi2} to impulsive systems without Lipschitz continuity of the jump map. In addition, Lemma~\ref{lem:ggi2} is not a particular case of other existing comparison-type results \citep[such as those in][]{nornes_auto14, lakbai_book89} and is hence interesting in its own right.
\begin{rem}
  \label{rem:hj0}
  If the function $a(\cdot)$ 
  is constant, it follows that $h_j^{t_0}(p,t) = h_j^0(p,t-t_0)$ for all $j \in \N_0$, $p\ge 0$ and $t\ge t_0 \ge 0$. \mer
\end{rem}

The following result is a generalization of Lemma~3 of \citet{haiman_tac18} to the current setting. The proof is given in Section~\ref{sec:proof-lemma-genlem3}.
\begin{lema}
  \label{lem:genlem3}
  Let $\{\Sigma_{\lambda}=(\gamma_{\lambda},f_{\lambda},g_{\lambda})\}_{\lambda \in \Lambda}$ be a strongly 0-GUAS parametrized family of impulsive systems with inputs which satisfies Assumptions \ref{as:ass-f-1} and \ref{as:ass-g-1}. Let $\beta\in\KL$ characterize the 0-GUAS property and let $\nu_f$ and $\nu_g$ be the functions given by Assumptions \ref{as:ass-f-1}.\ref{item:fbound}) and \ref{as:ass-g-1}.\ref{item:gbound}). Let $\chi_f,\chi_g \in \Ki$ satisfy $\chi_f \ge \nu_f$ and $\chi_g \ge \nu_g$. Then, for every $r>0$ and every $\eta>0$, there exist $L=L(r)$, $\kappa = \kappa(r,\eta)$ and $\omega=\omega_r\in \ki$ such that if $x\in\T_{\Sigma_{\lambda}}(t_0,x_0,u)$ with $\lambda\in \Lambda$, $t_0 \ge 0$, $x_0\in\R^n$ and $u\in \U$ satisfies $|x(t)| \le r$ for all $t\ge t_0$, then also
  \begin{multline}
    \label{eq:genlem3bnd}
    |x(t)| \le \beta(|x_0|,t-t_0+n^{\gamma_{\lambda}}_{(t_0,t]}) +\\
    h_{n^{\gamma_{\lambda}}_{(t_0,t]}}^0 \left ((t-t_0+n^{\gamma_{\lambda}}_{(t_0,t]})\eta + \kappa \|u_{(t_0,t]}\|_{\chi_f,\chi_g,\gamma_{\lambda}},t-t_0 \right),
  \end{multline}
  where $h_j^0$, for $j=0,1,\ldots$, are the functions defined in Lemma~\ref{lem:ggi2} in correspondence with $a(s) \equiv L$ and $c_j \equiv 1$.  
\end{lema}
%
As in \citet[Lemma~3]{haiman_tac18}, the inequality (\ref{eq:genlem3bnd}) is only useful when its right-hand side is less than $r$, since $|x(t)| \le r$ for all $t\ge t_0$ is already assumed. If $\gamma_\lambda=\emptyset$ (no impulses), and hence $n^{\gamma_{\lambda}}_{(t_0,t]} = 0$, then (\ref{eq:genlem3bnd}) reduces to the corresponding bound in Lemma~3 of \citet{haiman_tac18}.

The following result shows that if a system is strongly 0-GUAS, then UBEBS could be equivalently defined setting $c=0$ in (\ref{eq:cubebs}). This generalizes Lemma~4 of \citet{haiman_tac18} to the current setting. The proof is given in Section~\ref{sec:proof-lemma-0UBEBS}.
\begin{lema}
  \label{lem:0UBEBS} 
  Let $\{\Sigma_{\lambda}=(\gamma_{\lambda},f_{\lambda},g_{\lambda})\}_{\lambda \in \Lambda}$ be a strongly 0-GUAS and UBEBS parametrized family of impulsive systems with inputs which satisfies Assumptions \ref{as:ass-f-1} and \ref{as:ass-g-1}. Then there exist $\tilde\alpha,\tilde\rho_1,\tilde\rho_2\in \Ki$, with $\tilde\rho_1 \ge \nu_f$ and $\tilde\rho_2 \ge \nu_g$, for which the estimate (\ref{eq:0UBEBS}) holds for every $x \in \T_{\Sigma_{\lambda}}(t_0,x_0,u)$ with $\lambda \in \Lambda$, $t_0\ge 0$, $x_0\in \R^n$ and $u\in \U$.
  \begin{align} 
    \label{eq:0UBEBS}
    \tilde\alpha(|x(t)|) \le |x(t_0)| + \|u_{(t_0,t]}\|_{\tilde\rho_1,\tilde\rho_2,\gamma_{\lambda}} \quad \forall t\ge t_0.
 \end{align}
\end{lema}

We now have almost all the ingredients required for proving Theorem \ref{UBEBSand0GASiffiiss}. The only additional step is an $\epsilon$-$\delta$ characterization of the strong iISS property \citep[see][Theorem 3.2]{haiman_rpic19}, stated here so that iISS is uniform over families of systems.
\begin{teo}
  \label{thm:eps-delta}
 Consider the parametrized family  $\Sigma_\Lambda = \{\Sigma_{\lambda} = (\gamma_{\lambda}, f_{\lambda}, g_{\lambda}) \}_{\lambda \in \Lambda}$ of impulsive systems with inputs. Let $\rho_1,\rho_2\in\Ki$. Consider the notation $\|u\|_{\lambda} = \|u\|_{\rho_1,\rho_2,\gamma_{\lambda}}$ and, for $r \ge 0$, $B_r^{\lambda} := \{u\in\U: \| u \|_{\lambda} \le r\}$. Then $\Sigma_\Lambda$ is strongly iISS with gain $(\rho_1,\rho_2)$ if and only if the following conditions hold:
  \begin{enumerate}[i)]
  \item For every $T\ge 0$, $r\ge 0$, $s\ge 0$, there exists $C>0$ such that every $x\in\T_{\Sigma_{\lambda}}(t_0,x_0,u)$ with $\lambda \in \Lambda$, $t_0 \ge 0$, $x_0 \in B_r^n$ and $u \in B_s^{\lambda}$ satisfies $|x(t)| \le C$ for all $t\ge t_0$ such that $t+n^{\gamma_{\lambda}}_{(t_0,t]}\le t_0+T$.\label{item:fc}
  \item For each $\epsilon > 0$, there exists $\delta > 0$ such that every $x\in\T_{\Sigma_{\lambda}}(t_0,x_0,u)$ with $\lambda \in \Lambda$, $t_0 \ge 0$, $x_0 \in B_\delta^n$ and $u \in B_\delta^{\lambda}$ satisfies $|x(t)| \le \epsilon$ for all $t\ge t_0$.\label{item:sisicss}
  \item There exists $\alpha\in\Ki$ such that for every $r,\epsilon > 0$ there exists $T>0$ so that \label{item:gatt}
    \begin{align*}
      \alpha(|x(t)|) &\le \epsilon + \|u\|_{\lambda}
    \end{align*}
for all $x\in\T_{\Sigma_{\lambda}}(t_0,x_0,u)$, $\lambda \in \Lambda$, $t_0 \ge 0$, $x_0 \in B_r^n$, $u \in \U$, and $t\ge t_0$ such that $t + n^{\gamma_{\lambda}}_{(t_0,t]} \ge t_0+T$.
  \end{enumerate}
\end{teo}

\subsection{Proof of Theorem \ref{UBEBSand0GASiffiiss}}
\label{sec:proof-iiss-charact}

 ($\Rightarrow$) Let $x\in \T_{\Sigma_{\lambda}}(t_0,x_0,\mathbf{0})$, with $\lambda\in \Lambda$, $t_0\ge 0$, $x_0\in \R^n$. The estimate (\ref{eq:ciiss}), with $\gamma_{\lambda}$ instead of $\gamma$ reduces to $\alpha(|x(t)|) \le \beta(|x(t_0)|,t-t_0+n^{\gamma_{\lambda}}_{(t_0,t]})$ and hence $|x(t)| \le \alpha^{-1}(\beta(|x(t_0)|,t-t_0+n^{\gamma_{\lambda}}_{(t_0,t]})$. The function $\tilde\beta := \alpha^{-1} \comp \beta$ satisfies $\tilde\beta \in \KL$, and hence (\ref{eq:0-guas}) follows with $\beta$ replaced by $\tilde\beta$. Therefore, clearly strong iISS implies strong 0-GUAS. Consider $\beta \in \KL$ from (\ref{eq:ciiss}), define $\beta_0 \in \Ki$ via $\beta_0(r) = \beta(r,0)$. Define $\psi\in\Ki$ via $\psi(r) := \min\{\beta_0^{-1}(r/2), r/2\}$. Applying $\psi$ to each side of (\ref{eq:ciiss}), we obtain
 \begin{align*}
   \psi\comp\alpha(|x(t)|) &\le \psi\left( \beta_0(|x(t_0)|) + \|u_{(t_0,t]}\|_{\rho_1,\rho_2,\gamma_{\lambda}} \right )\\
                           &\le \psi\left(2\beta_0(|x(t_0)|) \right) + \psi(2 \|u_{(t_0,t]}\|_{\rho_1,\rho_2,\gamma_{\lambda}})\\
                           &\le |x(t_0)| + \|u_{(t_0,t]}\|_{\rho_1,\rho_2,\gamma_{\lambda}},
 \end{align*}
 and hence (\ref{eq:cubebs}) follows with $\alpha$ replaced by $\tilde\alpha := \psi\comp\alpha \in \Ki$. We have thus shown that strong iISS implies UBEBS.
 
 ($\Leftarrow$)
 Let $\tilde\alpha,\tilde\rho_1,\tilde\rho_2\in\Ki$ be given by Lemma~\ref{lem:0UBEBS}, so that (\ref{eq:0UBEBS}) is satisfied.
 We will prove that $\{\Sigma_{\lambda}\}_{\lambda \in \Lambda}$ is strongly iISS with iISS gain $(\tilde\rho_1,\tilde\rho_2)$ by establishing each of the items of Theorem~\ref{thm:eps-delta}.  Here we use the notation $\|u\|_{\lambda} = \|u\|_{\tilde\rho_1,\tilde\rho_2,\gamma_{\lambda}}$.
 
 \ref{item:fc}) Let $T\ge 0$, $r\ge 0$ and $s\ge 0$. Let $x \in \T_{\Sigma_{\lambda}}(t_0,x_0,u)$ with $\lambda \in \Lambda$, $t_0\ge 0$, $x_0\in B_r^n$ and $u\in B^{\lambda}_s$. From (\ref{eq:0UBEBS}), it follows that $\tilde\alpha(|x(t)|) \le r+s$, and hence $|x(t)| \le \tilde\alpha^{-1}(r+s) =: C$ for all $t\ge t_0$. This establishes item~\ref{item:fc}) of Theorem~\ref{thm:eps-delta}.
 
 \ref{item:sisicss}) Let $\epsilon > 0$ and $\delta = \tilde\alpha(\epsilon)/2$. Then, if $x \in \T_{\Sigma_{\lambda}}(t_0,x_0,u)$ with $\lambda \in \Lambda$, $t_0\ge 0$, $x_0\in B_\delta^n$ and $u\in B_\delta^\lambda$, it follows from (\ref{eq:0UBEBS}) that $|x(t)|\le \tilde\alpha^{-1}(2\delta) = \epsilon$ for all $t\ge t_0$. This establishes item~\ref{item:sisicss}) of Theorem~\ref{thm:eps-delta}.
 
 \ref{item:gatt}) Let $\alpha = \tilde\alpha/2 \in\Ki$. Let $r,\epsilon > 0$ and let $x \in \T_{\Sigma_{\lambda}}(t_0,x_0,u)$ with $\lambda \in \Lambda$, $t_0\ge 0$, $x_0\in B_r^n$ and $u\in \U$. 
 We distinguish two cases:
 \begin{enumerate}[(a)]
 \item $\|u\|_{\lambda} \ge r$,
 \item $\|u\|_{\lambda} < r$.
 \end{enumerate}
 In case (a), from (\ref{eq:0UBEBS}) we have $\tilde\alpha(|x(t)|) \le r + \|u_{(t_0,t]}\|_{\lambda} \le r + \|u\|_{\lambda}\le 2\|u\|_{\lambda}$, hence $\alpha(|x(t)|) \le \|u\|_{\lambda} \le \epsilon + \|u\|_{\lambda}$ for all $t\ge t_0$.
 
 Next, consider case (b). From (\ref{eq:0UBEBS}), we have  $\tilde\alpha(|x(t)|) \le r + \|u\|_{\lambda} < 2r$ for all $t\ge t_0$. Then $|x(t)|\le \tilde r:={\tilde{\alpha}}^{-1}(2r)$ for all $t\ge t_0$. Let $\beta \in \KL$ characterize the strong 0-GUAS property, so that (\ref{eq:0-guas}) is satisfied under zero input, and let $L = L(\tilde r) > 0$ and $\omega=\omega_{\tilde r}\in \Ki$ be given by Lemma~\ref{lem:genlem3} with $\chi_f = \tilde\rho_1$ and $\chi_g = \tilde\rho_2$, and let $h_j:= h_j^0$, $j=0,1,\ldots$, be the functions defined in Lemma \ref{lem:ggi2} in correspondence with $a(s) \equiv L$ and $c_j \equiv 1$. Let $\tilde\epsilon = \epsilon$ and $\tilde T>0$ satisfy $\beta(\tilde r,\tilde T)< \tilde\epsilon /2$. Let $\tilde k= \lceil \tilde T \rceil +1$, where $\lceil s\rceil$ denotes the least integer not less than $s\in \R$. Since $h_{\tilde k}$ is continuous and $h_{\tilde k}(0,t)=0$ for all $t\ge 0$, then there exists $\tilde \delta>0$ such that $h_{\tilde k}(\tilde \delta,\tilde k)<\tilde \epsilon/2$. Define $\eta= \frac{\tilde \delta}{2 \tilde k}$ and let $\kappa = \kappa(\tilde r,\eta) > 0$ be given by Lemma~\ref{lem:genlem3}. Set $\delta = \frac{\tilde \delta}{2 \kappa}$ and define $N := \left\lceil \frac{r}{\delta} \right\rceil$ and  $T:= N \tilde k$.
 
 Consider the sequence $t_0=s_0<s_1<\ldots<s_N$, recursively defined as follows: 
 $$s_j=\inf\{t\ge s_{j-1}:t-s_{j-1}+n^{\gamma_{\lambda}}_{(s_{j-1},t]}\ge \tilde T\}.$$
  Consider the intervals $I_i=(s_{i}, s_{i+1}]$, with $i=0,\ldots, N-1$. We claim that there exists $j\le N-1$ for which $\|u_{I_j}\|_{\lambda} \le \delta$. For a contradiction, suppose that $\|u_{I_j}\|_{\lambda} > \delta$ for all $0 \le j \le N-1$. Then, $\|u\|_{\lambda} \ge  \sum_{j=0}^{N-1} \|u_{I_j}\|_{\lambda}> N\delta \ge r$, contradicting case (b). Therefore, let $0\le j\le N-1$ be such that $\|u_{I_j}\| \le \delta$. 
 
 Since $x\in \T_{\Sigma_{\lambda}}(s_{j},x(s_j),u)$ and $|x(t)|\le \tilde r$ for all $t\ge s_j$, from Lemma~\ref{lem:genlem3} it follows that
   \begin{multline*}
     |x(s_{j+1})| \le \beta\left( |x(s_j)|,s_{j+1}-s_{j}+n^{\gamma_{\lambda}}_{I_j} \right)+\\
     h_{n^{\gamma_{\lambda}}_{I_j}}\left( \left[s_{j+1}-s_j+n^{\gamma_{\lambda}}_{I_j}\right] \eta + \kappa \|u_{I_j}\|_{\lambda},\ s_{j+1}-s_j \right).
   \end{multline*}
   Since $|x(s_j)|\le \tilde r$, $\tilde T\le s_{j+1}-s_{j}+n^{\gamma_{\lambda}}_{I_j}\le \tilde k$,  $\tilde k \eta =\tilde \delta/2$, $\kappa \delta \le \tilde \delta/2$ and the functions $h_j(p,t)$ are separately increasing in $p$ and in $t$, and $h_j(p,t)\le h_{\tilde k}(p,t)$ for all $0\le j\le \tilde k$, it follows that
    \begin{align*}
   |x(s_{j+1})| \le \beta(\tilde r,\tilde T)+
   h_{\tilde k}(\tilde \delta, \tilde k)<\frac{\tilde \epsilon}{2}+\frac{\tilde \epsilon}{2}=\tilde \epsilon.
   \end{align*}
 Therefore, using (\ref{eq:0UBEBS}) with $t_0$ replaced by $s_{j+1}$, we reach
 \begin{align*}
   \tilde\alpha(|x(t)|) \le |x(s_{j+1})| + \| u_{(s_{j+1},t]} \|_{\lambda} \le \tilde\epsilon + \|u\|_{\lambda}
 \end{align*}
 for all  $t\ge s_{j+1}$. Since $s_{i+1}-s_{i}+n^{\gamma_{\lambda}}_{(s_{i},s_{i+1}]}\le \tilde k$ for all $0\le i \le N-1$, $s_{j+1}-t_0+n^{\gamma_{\lambda}}_{(t_0,s_{j+1}]}=\sum_{i=1}^j [s_{i+1}-s_{i}+n^{\gamma_{\lambda}}_{(s_{i},s_{i+1}]}]\le N\tilde k=T$. In consequence, if $t\ge t_0$ is such that $t-t_0+n^{\gamma_{\lambda}}_{(t_0,t]}\ge T$, then $t\ge s_{j+1}$, and hence $\tilde\alpha(|x(t)|)\le \tilde\epsilon + \|u\|_{\lambda}$. Since $\alpha =\tilde\alpha/2 \le \tilde\alpha$, it follows that item~\ref{item:gatt}) of Theorem~\ref{thm:eps-delta} also is satisfied.

\section{Strong ISS implies strong iISS}
\label{sec:iss-iiss}

To establish that strong ISS implies strong iISS, we need Assumption~\ref{as:strongcont}, which strengthens Assumptions~\ref{as:ass-f-1} and~\ref{as:ass-g-1}. 
\begin{as}
  \label{as:strongcont}
  The functions $f_{\lambda}$ satisfy B\ref{item:fdeltau})--B\ref{item:fLip}) with the subscript `$a$' replaced by `$f$'. The functions $g_\lambda$ satisfy B\ref{item:fdeltau})--B\ref{item:fdeltax}) with $f_\lambda$ replaced by $g_\lambda$ and the subscript `$a$' replaced by `$g$':
  \begin{enumerate}[B1)]
  \item There exists $\tilde\varphi_a \in \Ki$ and nondecreasing and continuous functions $N_a,O_a : \R_{\ge 0} \to \R_{\ge 0}$ such that\footnote{Recall the notation $a\wedge b = \min\{a,b\}$.}
    \begin{multline*}
      |f_{\lambda}(t,\xi,\mu_1) - f_{\lambda}(t,\xi,\mu_2)|\\ \le
      \tilde\varphi_a(|\mu_1 - \mu_2|) \big[N_a(|\xi|) + O_a\big( |\mu_1| \wedge |\mu_2|\big) \big]
    \end{multline*}
    holds for all $t\ge 0$, $\xi\in\R^n$, $\mu_1,\mu_2 \in \R^m$ and $\lambda\in\Lambda$. \label{item:fdeltau}
  \item There exist $\eta_a,\varphi_a \in \Ki$, and $P_a : \R_{\ge 0} \to \R_{\ge 0}$ nondecreasing and continuous, such that \label{item:fdeltax}
    for all $t\ge 0$, $\xi_1,\xi_2\in\R^n$, $\mu \in \R^m$, and $\lambda\in\Lambda$,
    \begin{multline*}
      |f_{\lambda}(t,\xi_1,\mu) - f_{\lambda}(t,\xi_2,\mu)|\\ \le
      \eta_a(|\xi_1 - \xi_2|) [ P_a(|\xi_1|\wedge |\xi_2|) + \varphi_a(|\mu|) ].
    \end{multline*}
  \item Item B\ref{item:fdeltax}) holds, in addition, with $\eta_a$ such that for every $M\ge 0$ there exists $L^f = L^f(M)$ so that 
    \begin{align}
      \label{eq:etaLip}
      \eta_a(s) &\le L^f s\quad \text{for all }0\le s \le M,
    \end{align}
    where the function $L^f(\cdot)$ is continuous, nondecreasing, and positive for $M>0$. \label{item:fLip}
  \end{enumerate}
\end{as}

It is easy to show that Assumptions \ref{as:ass-f-1} and \ref{as:ass-g-1} follow from Assumption \ref{as:strongcont} and the blanket assumptions $f_{\lambda}(t,0,0)=0$ and $g_{\lambda}(t,0,0)=0$ for all $t\ge 0$ and $\lambda \in \Lambda$.

The conditions in Assumption \ref{as:strongcont} are equivalent to the following set of conditions, which are similar to those considered in Assumption~1 of \citet{haiman_auto19} (see also Lemma~3.4 therein).
  \begin{enumerate}[{A}1)]
  \item There exists $\omega_1\in\Ki$ and for every $r,s\ge 0$, there exists $L_1 = L_1(r,s) \ge 0$ such that
    \begin{align}
      |f_{\lambda}(t,\xi,\mu_1) - f_{\lambda}(t,\xi,\mu_2)| &\le L_1 \omega_1(|\mu_1 - \mu_2|)
    \end{align}
    for all $t\ge 0$, $\xi \in B_r^n$, $\mu_1,\mu_2\in B_s^m$ and $\lambda\in\Lambda$.\label{item:req-contu}
  \item There exists $\omega_2\in\Ki$ and for every $r,s\ge 0$, there exists $L_2 = L_2(r,s) \ge 0$ such that
    \begin{align}
      |f_{\lambda}(t,\xi_1,\mu) - f_{\lambda}(t,\xi_2,\mu)| &\le L_2 \omega_2(|\xi_1 - \xi_2|)
    \end{align}
    for all $t\ge 0$, $\xi_1,\xi_2 \in B_r^n$, $\mu\in B_s^m$ and $\lambda\in\Lambda$.\label{item:req-contx}
  \item Item A\ref{item:req-contx}) holds with $\omega_2(r)\equiv r$. \label{item:req-Lipx}
  \end{enumerate}
By proceeding as in the proof of Lemma~3.4 in \citet{haiman_auto19}, it follows that B\ref{item:fdeltau})--B\ref{item:fdeltax}) are equivalent to A\ref{item:req-contu})--A\ref{item:req-contx}) and that B\ref{item:fdeltau})--B\ref{item:fLip}) are equivalent to A\ref{item:req-contu})--A\ref{item:req-Lipx}). It is thus clear that Assumption~\ref{as:strongcont} imposes local Lipschitz continuity of the flow maps and hence uniqueness of solutions.

Our main result is the following.
\begin{teo}
  \label{thm:issimpliesiiss}
  Let $\{\Sigma_\lambda = (\gamma_\lambda,f_\lambda,g_\lambda)\}_{\lambda\in\Lambda}$ be a strongly ISS parametrized family of impulsive systems with inputs and let Assumption~\ref{as:strongcont} hold. Then, $\{\Sigma_\lambda\}_{\lambda\in\Lambda}$ is strongly iISS.
\end{teo}

The structure of the proof of Theorem~\ref{thm:issimpliesiiss} is given in the following diagram. In this diagram, the application of Theorem~\ref{UBEBSand0GASiffiiss} is possible because Assumption~\ref{as:strongcont} implies Assumptions~\ref{as:ass-f-1} and~\ref{as:ass-g-1}.
\begin{equation*}
  \xymatrix@R-\baselineskip@C+4mm{ & \smash{\stackrel{\text{strong}}{\text{0-GUAS}}}\ar[d] \\
    \stackrel{\text{strong}}{\text{ISS}}\ar@2{->}^{\text{Remark}~\ref{rem:sISS_s0-GUAS}\ }[ru]\ar@2{->}_{\tiny{\text{Thm}~\ref{thm:iss-implies-ubebs}}}[rd] & \text{and} \ar@2{->}^{\text{Thm}~\ref{UBEBSand0GASiffiiss}}[r] & \smash{\stackrel{\text{strong}}{\text{iISS}}} \\
    & \ar[u]\text{UBEBS} & }
  \end{equation*}
Before giving the remaining step, indicated as Theorem~\ref{thm:iss-implies-ubebs}, we pose the following simple consequence of Theorem~\ref{thm:issimpliesiiss} and Proposition \ref{prop:weakstrong}.
\begin{cor}
  Let $\{\Sigma_\lambda = (\gamma_\lambda,f_\lambda,g_\lambda)\}_{\lambda\in\Lambda}$ be a weakly ISS parametrized family of impulsive systems with inputs and let Assumption~\ref{as:strongcont} hold. Suppose that $\{\gamma_\lambda\}_{\lambda\in\Lambda}$ is UIB (Definition~\ref{def:UIB}). Then, $\{\Sigma_\lambda\}_{\lambda\in\Lambda}$ is strongly iISS and hence also weakly iISS.
\end{cor}
\begin{proof}
  Since $\{\gamma_\lambda\}_{\lambda\in\Lambda}$ is UIB, then by Proposition~\ref{prop:weakstrong} the weak and strong versions of ISS (or iISS) are equivalent. Applying Theorem~\ref{thm:issimpliesiiss}, the result follows.\qed
\end{proof}

We next give a theorem that establishes that strong ISS implies UBEBS. This theorem is an extension to impulsive systems of Theorem~3.12 in \citet{haiman_auto19}. However, due to the absence of any type of Lipschitz continuity assumption on the jump maps, the current proof does not follow straightforwardly from the corresponding one in \citet{haiman_auto19}. Moreover, the proof is not a simple consequence of replacing the application of Gronwall inequality by that of the current Lemma~\ref{lem:ggi2}. Specifically, the expression to be bounded does not anymore have the multiplicative form given as $g_1(r)g_2(s)$ in Lemma~3.11 of \citet{haiman_auto19}, leading to a novel bounding strategy.%
\begin{teo}
  \label{thm:iss-implies-ubebs}
  Let $\{\Sigma_\lambda = (\gamma_\lambda,f_\lambda,g_\lambda)\}_{\lambda\in\Lambda}$ be a strongly ISS parametrized family of impulsive systems with inputs and let Assumption~\ref{as:strongcont} hold. Then, $\{\Sigma_\lambda\}_{\lambda\in\Lambda}$ is UBEBS.
\end{teo}
\begin{proof}
  Let $\beta\in\KL$ and $\rho\in\Ki$ characterize the strong ISS property. Define $h_1^a,h_2^a:\R_{\ge 0}^2 \to \R$ via
  \begin{align}
    \label{eq:hdef}
    h_1^a(r,b) &:= N_a(\beta(r,0) + \rho(b)) + O_a(b),\\
    h_2^a(r,b) &:= P_a(\beta(r,0) + \rho(b)),
  \end{align}
  where $a \in \{f,g\}$. Let $L^f : \R_{\ge 0} \to \R_{\ge 0}$ be continuous, nondecreasing, and such that for every $M\ge 0$, (\ref{eq:etaLip}) holds with the subscript `$a$' replaced by `$f$' and $L^f = L^f(M)$. In correspondence with every $r> 0$, define $T_r > 1$ continuous and such that
  \begin{align}
    \label{eq:Tr}
    \beta(r,T_r&-1) \le r/3,\quad\text{and also}\\
    \label{eq:brMr}
    b_r &:= \rho^{-1}(r/3), & M_r &:= r/3,\\
    \bar h_1(r) &:= h_1^f(r,b_r) + h_1^g(r,b_r) & L_r^f &:= L^f(M_r). 
  \end{align}
For each $j\in \N_0$, consider the functions $\tilde h_j : \R_{\ge 0}^4 \to \R_{\ge 0}$ given by 
  \begin{align*}
    \tilde h_0(p,&T,r,s) = pe^{[h_2^f(r,b_r)T + s] L_r^f},\quad \text{and for }j\ge 1,\\
    \tilde h_j(p,&T,r,s) = \tilde h_{j-1}(p,T,r,s)+ \\ 
                        &[h_2^g(r,b_r) + s] e^{[h_2^f(r,b_r)T + s] L_r^f}  \eta_g(\tilde h_{j-1}(p,T,r,s)),
  \end{align*}
  and define, for $r>0$ and $s\ge 0$,
  \begin{multline*}
    \tilde p(r,s) := \sup \Big\{ p \ge 0 : \tilde h_j(p,T,r,s) \le \frac{M_r}{2},\\ \forall (j,T) \text{ s.t. }T\ge 0, T+j \le T_r \Big\}.
  \end{multline*}
  Note that the functions $\tilde h_j$ are nondecreasing in $j$, $p$, $T$, $r$ and $s$, continuous in $(p,T,r,s)$ over $\R_{\ge 0}^4$, and satisfy $\tilde h_j(0,T,r,s) = 0$ for all $j\in\N_0$ and $(T,r,s) \in \R_{\ge 0}^3$. In addition, the function $\tilde h_j(\cdot,T,r,s)$ is increasing for every $j\in\N_0$ and $(T,r,s)\in\R_{\ge 0}^3$, and $\tilde h_j(p,T,r,\cdot)$ is increasing whenever $p>0$ and $r>0$. These facts make $\tilde p(r,s) > 0$ for all $r>0$ and $s\ge 0$, and $\tilde p(r,\cdot)$ decreasing.
%
From the definition of $\tilde p(r,s)$, we have that for all $r>0$ and $s\ge 0$,
  \begin{multline}
    \label{eq:htjbnd}
    \tilde h_j(p,T,r,s) \le M_r/2\\
    \text{whenever } p\le \tilde p(r,s), T\ge 0, T + j \le T_r.
  \end{multline}
%
  Consider the function $\ell : [1,\infty) \to \R_{\ge 0}$, defined via
  \begin{align}
    \label{eq:elldef}
    \ell(\bar r) &:= \sup_{1 \le r \le \bar r} \frac{\bar h_1(r) (r-1)}{\tilde p(r,r-1)}.
  \end{align}
  It is clear that $\ell$ is nondecreasing. 
  \begin{claim}
    \label{clm:ell}
    $\ell(\bar r) < \infty$ for all $\bar r \ge 1$.
  \end{claim}
  \textbf{Proof of Claim~\ref{clm:ell}:} Let $\bar r \ge 1$ and consider
  \begin{align*}
    \bar T := \sup_{1\le r \le \bar r} T_r,\\
    \bar p := \sup \Big\{ p\ge 0 &: \max_{j\in\N_0,j\le \bar T} 
    \tilde h_j(p,\bar T,\bar r,\bar r - 1)\le \frac{M_1}{2} \Big\}.
  \end{align*}
  Since $T_r$ is positive and continuous for $r>0$, then $\bar T$ is finite and positive. From the continuity and monotonicity properties of $\tilde h_j$, it follows that $\bar p > 0$. From the corresponding definitions, it also follows that $\tilde p(r,r-1) \ge \bar p$ for all $1\le r\le \bar r$. In consequence, by also taking into account the continuity of $\bar h_1$ it follows that
  \begin{align*}
  \bar \ell(\bar r) \le \max_{1\le r\le \bar r}\frac{\bar h_1(r)(r-1)}{\bar p}<\infty.
  \end{align*}\mer

It follows that there exists $\kappa \in \Ki$ such that $\ell(r) \le \kappa(r)$ for all $r\ge 1$. Define $\alpha\in\Ki$ via 
\begin{align}
\label{eq:alpha}
\alpha(b) = \kappa(3\rho(b)).
\end{align}
  Given an input $u\in\U$ and a constant $b\ge 0$, let $u_b$ denote a new input, defined as follows
  \begin{align}
    \label{eq:Phiub}
    u_b(t)&=
                    \begin{cases}
                      \dfrac{bu(t)}{|u(t)|} &\text{if }t\in\Omega_u(b),\\
                      u(t) &\text{otherwise,}
              \end{cases}\\
    \label{eq:Omega}
    \Omega_{u}(b) &:= \{ t\ge 0 :  |u(t)| > b \}.
  \end{align}
  Note that $|u_b(t)|= \min\{|u(t)|,b\}$ for all $t\ge 0$ and hence $\|u_b\|_{\infty,\gamma}\le b$ for all $\gamma$.

  Let $\chi_1,\chi_2 \in \Ki$ satisfy $\chi_1 \ge \max\{\varphi_f,\tilde\varphi_f^2,\alpha^2\}$ and $\chi_2 \ge \max\{\varphi_g,\tilde\varphi_g^2,\alpha^2\}$. We will establish UBEBS with gain $(\chi_1,\chi_2)$. Let $t_0 \ge 0$, $\xi\in\R^n$, $\lambda\in\Lambda$, set $\gamma=\gamma_\lambda$, and consider an input $u\in\U$ such that
  \begin{align}
    \label{eq:u1nrg}
    E := \int_0^\infty \chi_1(|u(s)|) ds + \sum_{s\in\gamma} \chi_2(|u(s)|) < \infty.
  \end{align}
  Let $x \in \T_{\Sigma_\lambda}(t_0,\xi,u)$ and define $\tilde\alpha \in \Ki$ via
  \begin{align}
    \label{eq:alptildef}
    \tilde\alpha(r) &= \beta(r,0) + \frac{2r}{3}.
  \end{align}
  \begin{claim}
    \label{clm:iter}
    Let $r$ be any real number such that $r\ge 1 + E$ and $|x(t_0)|\le r$, then
    \begin{align}
      |x(t)| \le \tilde\alpha(r) \quad \forall t\ge t_0.
    \end{align}
  \end{claim}
  \textbf{Proof of Claim~\ref{clm:iter}:} 
  For a fixed $b\ge 0$, let $x_b \in \T_{\Sigma_\lambda}(t_0, x(t_0), u_b)$, and $\Delta x = x - x_b$. From the strong ISS property, then
  \begin{align*}
    |x_b(t)| &\le \beta\left(|x(t_0)|,t-t_0+n^\gamma_{(t_0,t]}\right) + \rho(\|u_b\|_{\infty,\gamma})\\
    &\le \beta(r,0) + \rho(b)
  \end{align*}
  for all $t\ge t_0$. From (\ref{eq:solintform}) and Assumption~\ref{as:strongcont}, it follows that
  \begin{align*}
    &|\Delta x(t)|
 \le \int_{t_0}^t \Big| f_\lambda(s,x(s),u(s)) - f_\lambda(s,x_b(s),u_b(s)) \Big| ds +\\ 
    &\sum_{\tau\in\gamma\cap(t_0,t]} \Big| g_\lambda(\tau,x(\tau^-),u(\tau)) - g_\lambda(\tau,x_b(\tau^-),u_b(\tau)) \Big| \\
 &\le \int_{t_0}^t \Big| f_\lambda(s,x(s),u(s)) - f_\lambda(s,x_b(s),u(s)) \Big| ds +\\
    &\sum_{\tau\in\gamma\cap(t_0,t]} \Big| g_\lambda(\tau,x(\tau^-),u(\tau)) - g_\lambda(\tau,x_b(\tau^-),u(\tau)) \Big| + \\
 &\phantom{\le} \int_{t_0}^t \Big| f_\lambda(s,x_b(s),u(s)) - f_\lambda(s,x_b(s),u_b(s)) \Big| ds + \\
    &\sum_{\tau\in\gamma\cap(t_0,t]} \Big| g_\lambda(\tau,x_b(\tau^-),u(\tau)) - g_\lambda(\tau,x_b(\tau^-),u_b(\tau)) \Big| \\
    &\le \int_{t_0}^t \eta_f(|\Delta x(s)|) {\scriptstyle [P_f(|x(s)|\wedge |x_b(s)|)+\varphi_f(|u(s)|)]} ds +\\
    &\sum_{\tau\in\gamma\cap(t_0,t]} \eta_g(|\Delta x(\tau^-)|) {\scriptstyle [P_g(|x(\tau^-)|\wedge |x_b(\tau^-)|)+\varphi_g(|u(\tau)|)]} +\\
    &\int_{t_0}^t \tilde\varphi_f(|u(s) -u_b(s)|) {\scriptstyle [N_f(|x_b(s)|) +  O_f(|u(s)|\wedge |u_b(s)|)]} ds +\\
    &\sum_{\tau\in\gamma\cap(t_0,t]} \tilde\varphi_g(|u(\tau) -u_b(\tau)|) {\scriptstyle [N_g(|x_b(\tau^-)|) +  O_g(|u(\tau)|\wedge |u_b(\tau)|)]}
  \end{align*}
  holds for all $t\ge t_0$ for which $x(t)$ exists.  Then, for all $t\ge t_0$ for which $x(t)$ exists,
  \begin{align}
    |\Delta &x(t)| \le \int_{t_0}^t \eta_f(|\Delta x(s)|) [h_2^f(r,b)+\varphi_f(|u(s)|)] ds \notag \\ 
    &+\sum_{\tau\in\gamma\cap(t_0,t]} \eta_g(|\Delta x(\tau^-)|) [h_2^g(r,b) + \varphi_g(|u(\tau)|)] \notag\\
    &+ h_1^f(r,b) \int_{t_0}^t \tilde\varphi_f(|u(s) -u_b(s)|) ds \notag \\
    \label{eq:Deltax1}
    &+ h_1^g(r,b) \sum_{\tau\in\gamma\cap(t_0,t]} \tilde\varphi_g(|u(\tau) -u_b(\tau)|).
  \end{align}
  For $t\ge t_0$, we have the following inequalities:
  \begin{align*}
    \int_{t_0}^t \tilde\varphi_f(|u(s) - u_b(s)|) ds 
    &\le \int_{\Omega_{u}(b)} \tilde\varphi_f(|u(s)|) ds,\\
    \sum_{\tau\in\gamma\cap(t_0,t]} \tilde\varphi_g(|u(\tau) -u_b(\tau)|) 
    &\le \sum_{\tau\in\gamma\cap\Omega_u(b)} \tilde\varphi_g(|u(\tau)|)
  \end{align*}
  Applying the Schwarz inequality, then
  \begin{align*}
    \int_{\Omega_{u}(b)} \tilde\varphi_f(|u(s)|) ds &\le |\Omega_{u}(b)|^{1/2} \sqrt{\int_{\Omega_{u}(b)} \tilde\varphi_f^2(|u(s)|) ds}\\
    &\le |\Omega_{u}(b)|^{1/2} \sqrt{E},\quad{\text{and likewise}}\\
    \sum_{\tau\in\gamma\cap\Omega_u(b)} \tilde\varphi_g(|u(\tau)|) 
    &\le \sqrt{\#[\gamma\cap\Omega_u(b)]} \sqrt{E},
  \end{align*}
  where we have used the facts that $\chi_1\ge\tilde\varphi_f^2$ and $\chi_2 \ge \tilde\varphi_g^2$, and where $|\Omega_u(b)|$ denotes the Lebesgue measure of the set $\Omega_u(b)$. 
  Also, we have
  \begin{align*}
    E &\ge \int_{\Omega_{u}(b)} \chi_1(|u(s)|) ds \ge |\Omega_{u}(b)| \chi_1(b), \quad\text{and}\\
    E &\ge \sum_{\tau\in\gamma\cap\Omega_u(b)} \chi_2(|u(\tau)|) \ge \#[\gamma\cap\Omega_u(b)] \chi_2(b),
  \end{align*}
  and hence
  \begin{align*}
    |\Omega_{u}(b)| &\le \frac{E}{\chi_1(b)},\quad \text{and} \quad
                      \#[\gamma\cap\Omega_u(b)] \le \frac{E}{\chi_2(b)} \quad \text{if }b>0.
  \end{align*}
  Combining the obtained inequalities, we reach, for $b>0$,
  \begin{align}
    \label{eq:gtilE1}
    \int_{t_0}^t \tilde\varphi_f(|u(s) - u_b(s)|) ds &\le \frac{E}{\sqrt{\chi_1(b)}} \le \frac{E}{\alpha(b)},\\
    \label{eq:gtilE2}
    \sum_{\tau\in\gamma\cap\Omega_u(b)} \tilde\varphi_g(|u(\tau) -u_b(\tau)|) &\le  \frac{E}{\sqrt{\chi_2(b)}} \le \frac{E}{\alpha(b)},
  \end{align}
  where we have used the facts that $\chi_1 \ge \alpha^2$ and $\chi_2\ge\alpha^2$. Let $b=b_r$. Define
  \begin{align*}
    \iota &:= \inf\{t\ge t_0 : |\Delta x(t)| \ge M_r \}.
  \end{align*}
  We next show that $\iota - t_0 + n^\gamma_{(t_0,\iota]} > T_r$. Suppose on the contrary that  $\iota - t_0 + n^\gamma_{(t_0,\iota]} \le T_r$. From the definition of $\iota$ and the continuity of $\Delta x$ from the right, 
we have $\Delta x(\iota)\ge M_r$ and
  \begin{align}
    \label{eq:Deltax2}
    |\Delta x(t)| &< M_r\quad\text{for all }t_0 \le t < \iota,
  \end{align}
  From~(\ref{eq:etaLip}), then $\eta_f(|\Delta x(t)|) \le L_r^f |\Delta x(t)|$ for all $t_0 \le t < \iota$. From~(\ref{eq:Deltax1}) and (\ref{eq:gtilE1})--(\ref{eq:gtilE2}), then for all $t_0 \le t \le \iota$, we have
  \begin{align}
    |\Delta x(t)| \le p &+ \int_{t_0}^t a(s) |\Delta x(s)| ds\notag\\
    \label{eq:Deltax3}
    &+ \sum_{\tau\in\gamma\cap(t_0,t]} c(\tau) \eta_g( |\Delta x(\tau^-)|),\\
    \text{with}\quad p &= \frac{\bar h_1(r) E}{\alpha(b_r)}=\frac{\bar h_1(r) E}{\kappa(r)},\notag\\
    a(s) &= [h_2^f(r,b_r) + \varphi_f(|u(s)|)] L_r^f,\notag\\
    c(\tau) &= [h_2^g(r,b_r) + \varphi_g(|u(\tau)|)].\notag
  \end{align}
  Note that (\ref{eq:Deltax3}) holds also at $t=\iota$ even if only (\ref{eq:Deltax2}) is true and it happens that $|\Delta x(\iota)| > M_r$. Applying Lemma~\ref{lem:ggi2} with $y(t)=\Delta x(t)$, $T=\iota$, $\sigma = \gamma \cap (t_0,\iota] = \{s_j\}_{j=1}^k$, with $k = n^\gamma_{(t_0,\iota]}$, $\{c_j\}_{k=1}^\infty$, with $c_j =c(s_j)$ for $1\le j\le k$ and $c_j=0$ for $j>k$ and $\omega= \eta_g$, it follows that $\Delta x$ must also satisfy
  \begin{align} \label{eq:deltaiota}
    |\Delta x(\iota)| &\le h_k^{t_0}(p,\iota)
  \end{align}
  with the functions $h_j^{t_0}$, $j\in \N_0$, as defined in Lemma~\ref{lem:ggi2}.
  \begin{claim} 
    \label{clm:h-tildeh}
    For all $p\ge 0$, $t\ge t_0$ and $0\le j\le k$,	
    \begin{align}\label{eq:claimth}
      h_j^{t_0}(p,t) &\le \tilde h_j(p,t-t_0,r,E).
    \end{align}
  \end{claim}
  \textbf{Proof of Claim~\ref{clm:h-tildeh}:} We prove the claim by induction on $j$. For $j=0$, we have that for all $t\ge t_0$
  \begin{align*}
    h_0^{t_0}(p,t)&=p e^{\int_{t_0}^t a(s) ds}\\ 
                  &\le pe^{[h_2^f(r,b_r)(t-t_0) + E] L_r^f}=\tilde h_0(p,t-t_0,r,E)
  \end{align*}
  since
  \begin{align}
    \int_{t_0}^t a(s) ds &= \left [h_2^f(r,b_r)(t-t_0) + \int_{t_0}^t\varphi_f(|u(s)|)ds \right ] L_r^f \nonumber\\
                         &\le [h_2^f(r,b_r)(t-t_0) + E] L_r^f \label{eq:inta}
  \end{align}
  because $\varphi_f \le \chi_1$ and $\|u\|_{\chi_1,\chi_2,\gamma}= E$.
 
  Suppose now that for some $0\le j<k$, (\ref{eq:claimth}) holds for all $t\ge t_0$. Then, from the definition of the function $h_{j+1}^{t_0}$, it follows that 
  \begin{align*}
    h_{j+1}^{t_0}(p,t)= h_{j}^{t_0}(p,t)+c_{j+1}
    \sup_{t_0\le s\le t}\left [ \eta_g(h_{j}^{t_0}(p,s))e^{\int_{s}^t a(\tau) d\tau}\right ].
  \end{align*}
  Since $c_{j+1}=c(s_{j+1})\le h_2^g(r,b_r)+E$, because $\varphi_g\le \chi_2$ and $\|u\|_{\chi_1,\chi_2,\gamma} = E$, and using (\ref{eq:inta}), the nonnegativity of $a$, the inductive hypothesis, and the fact that the functions $\eta_j$ and $\tilde h_j$ are nondecreasing in each of their arguments, it follows that 
  \begin{multline*}
    h_{j+1}^{t_0}(p,t) \le \tilde h_{j}(p,t-t_0,r,E) + [h_2^g(r,b_r)+E]\cdot\\
    e^{[h_2^f(r,b_r)(t-t_0) + E]L_r^f}\eta_g\big(\tilde h_{j}(p,t-t_0,r,E)\big) \\
    =\tilde h_{j+1}(p,t-t_0,r,E),
  \end{multline*}
  and the proof of the claim follows. \mer
 
From Claim \ref{clm:h-tildeh} it then follows that 
\begin{align*}
h_k^{t_0}(p,\iota) \le \tilde h_k(p,\iota-t_0,r,E).
\end{align*}
 
 On the other hand, for all $E \le r-1$, we have
  \begin{align*}
  p &= \frac{\bar h_1(r) E}{\kappa(r)} \le \frac{\bar h_1(r) (r-1)}{\kappa(r) } \le \frac{\bar h_1(r) (r-1)}{\kappa(r) }\frac{\tilde p(r,r-1)}{\tilde p(r,r-1)}\\ &\le \frac{\ell(r)}{\kappa(r)} \tilde p(r,E) \le \tilde p(r,E).
  \end{align*}  
 
 Therefore, since $\iota - t_0 + k \le T_r$, it follows from the definition of $\tilde p$ that $\tilde h_k(p,\iota-t_0,r,E)\le M_r/2$ and then, from (\ref{eq:deltaiota}) that
  \begin{align*}
    |\Delta x(\iota)| \le h^{t_0}_k(p,\iota) \le \tilde h_k(p,\iota-t_0,r,E) \le M_r/2,
  \end{align*}
 which is a contradiction. Thus $\iota - t_0 + n^\gamma_{(t_0,\iota]} > T_r$. Therefore, the solution $x$ can be bounded as follows
  \begin{align*}
    |x(t)| &\le |x_{b_r}(t)| + |\Delta x(t)| \\&\le \beta(r,t-t_0+n^\gamma_{(t_0,t]}) + \rho(b_r) + M_r\\
    &\le \beta(r,0) + \rho(b_r) + M_r = \tilde\alpha(r),
  \end{align*}
  for all $t\ge t_0$ such that $t - t_0 + n^\gamma_{(t_0,t]} \le T_r$. Consider the sequence $t_1 < t_2 < \cdots$, defined recursively as follows, for $j=0,1,2,\ldots$
  \begin{align*}
    t_{j+1} = \inf\{t > t_j : t-t_j + n^\gamma_{(t_j,t]} \ge T_r - 1 \}
  \end{align*}
  Note that $T_r - 1 \le t_{j+1} - t_j + n^\gamma_{(t_j,t_{j+1}]} \le T_r$, and that $t_j \to \infty$ because $\gamma$ has no finite limit points. It follows that
  \begin{align*}
    |x(t_1)| &\le \beta(r,T_r-1) + \rho(b_r) + M_r \le r. 
  \end{align*}
  Shifting the initial time to $t_i$ and applying recursively the preceding reasoning, we obtain
  \begin{align*}
   |x(t)|&\le \tilde \alpha(r)\quad \forall t\in [t_i,t_{i+1}]\\
   |x(t_{i+1})|&\le r.
  \end{align*}
  This concludes the proof of the claim. \mer

If $|x(t_0)|\ge 1 + E$, by applying Claim \ref{clm:iter} with $r=|x(t_0)|$ it follows that $|x(t)|\le \tilde \alpha(|x(t_0)|)$ for all $t\ge t_0$.

If $|x(t_0)|< 1+E$, let $t_1=\inf\{t\ge t_0:|x(t)|\ge 1+E\}$. If $t_1=\infty$, then $|x(t)|< 1+E$ for all $t\ge t_0$. If $t_1$ is finite, then $|x(t)|< 1+E$ for all $t\in [t_0,t_1)$. If $t_1 \notin \gamma$, then $|x(t_1)| = 1+E$. If $t_1 \in \gamma$, then $|x(t_1)| \le |x(t_1^-)| + |g_\lambda(t_1,x(t_1^-),u(t_1))|$. From B\ref{item:fdeltau}) in Assumption~\ref{as:strongcont} and the fact that $\chi_2 \ge \tilde\varphi_g$, it follows that $|g_\lambda(t_1,x(t_1^-),u(t_1)) - g_\lambda(t_1,x(t_1^-),0)| \le E [N_g(1+E) + O_g(0)]$ and from B\ref{item:fdeltax}), also $|g_\lambda(t_1,x(t_1^-),0) - g_\lambda(t_1,0,0)| \le \eta_g(1+E)P_g(0)$. Since in addition $g_\lambda(t_1,0,0) = 0$, then $|x(t_1)| \le (1+E) [1 + N_g(1+E) + O_g(0)] + \eta_g(1+E)P_g(0) =: \Psi(E)$, where $\Psi : \R_{\ge 0} \to \R_{\ge 0}$ is continuous and nondecreasing. 
By applying Claim \ref{clm:iter} with $t_1$ instead of $t_0$ and $r:=\Psi(E) \ge 1+E$ we obtain  $|x(t)|\le \tilde \alpha\comp\Psi(E)$ for all $t\ge t_1$. Therefore $|x(t)|\le \tilde \alpha\comp\Psi(E)$ for all $t\ge t_0$.
Since $\Psi$ is continuous and nondecreasing, there exists $\tilde\Psi\in\Ki$ such that $\Psi(r) \le \Psi(0) + \tilde\Psi(r)$ for all $r\ge 0$.
For all $t\ge t_0$ we have
\begin{align*}
|x(t)|& \le \max\{\tilde \alpha(|x(t_0)|),\tilde \alpha\comp\Psi(E)\}\\
& \le \tilde \alpha(|x(t_0)|)+\tilde \alpha\comp\Psi(E)\\
& \le \tilde \alpha(|x(t_0)|)+ \tilde\alpha(2\tilde\Psi(E)) + \tilde\alpha(2\Psi(0)),
\end{align*} 
where $\psi(\cdot) := \tilde\alpha(2\tilde\Psi(\cdot)) \in \Ki$. It thus follows that the family of impulsive systems is strongly UBEBS with UBEBS gain ($\chi_1,\chi_2$).\qed
\end{proof}

\section{Complementary proofs}
\label{sec:proofs}

\subsection{Proof of Lemma~\ref{lem:ggi2}}
\label{sec:proof-lem-ggi2}

For the sake of simplicity we write $h_j$ instead of $h_j^{t_0}$.

First, we prove that for all $t_0 \le r \le t$, it happens that $h_k(p,r)e^{\int_{r}^t a(s) ds} \le h_{k}(p,t)$ for all $k\in\N_0$. For $k=0$, we have
  \begin{align*}
    h_0(p,r)e^{\int_{r}^t a(s) ds} = p e^{\int_{t_0}^t a(s) ds} = h_0(p,t),
  \end{align*}
  so that the inequality holds with equality for $k=0$. Next, suppose that the inequality holds for some $k\in\N_0$. We have
  \begin{align*}
    h&_{k+1}(p,r)e^{\int_{r}^t a(s) ds}
 = e^{\int_{r}^t a(s) ds} \Big(h_{k}(p,r) + \\
&\phantom{==}c_{k+1} e^{\int_{t_0}^r a(s) ds} \sup_{t_0 \le s \le r} \left[ \omega(h_{k}(p,s))e^{-\int_{t_0}^s a(\tau) d\tau} \right]\Big) \\
    &= h_{k}(p,r) e^{\int_{r}^t a(s) ds} +\\
    &\phantom{==}c_{k+1} e^{\int_{t_0}^t a(s) ds} \sup_{t_0 \le s \le r} \left[ \omega(h_{k}(p,s))e^{-\int_{t_0}^s a(\tau) d\tau} \right]\\
    &\le h_k(p,t) + c_{k+1} e^{\int_{t_0}^t a(s) ds} \sup_{t_0 \le s \le t} \left[ \omega(h_{k}(p,s))e^{-\int_{t_0}^s a(\tau) d\tau} \right]\\
    &= h_{k+1}(p,t),
  \end{align*}
  so that the inequality holds for $k+1$. 

  Define $s_0 := t_0$ and recall that $\{s_k\}_{k=1}^N$, with $s_1 > s_0$, is the sequence of points where $y$ is discontinuous. Let $z:[t_0,T]\to \R_{\ge 0}$ be defined by
\begin{align}
 z(t)=p+\int_{t_0}^t a(s) y(s) ds+ \sum_{s\in \sigma \cap (t_0,t]} c(s) \omega(y(s^-)).
\end{align}
By assumption, $y(t)\le z(t)$ for all $t\in [t_0,T]$. We will prove by induction the following.

{\em Claim:} for all $0\le k\le N-1$, $z(t)\le h_k(p,t)$ for all $s_k\le t< s_{k+1}$.

{\em Case} $k=0$.
We have that for all $t\in [s_0,s_1)$, $\sigma \cap (t_0,t]=\emptyset$. Therefore, for all $t\in [s_0,s_1)$,
$$ z(t)= p+ \int_{t_0}^t a(s) y(s)ds\le p+ \int_{t_0}^t a(s) z(s) ds.$$
Applying Gronwall inequality, we have that 
$$ z(t)\le p e^{\int_{t_0}^t a(s) ds} = h_0(p,t)\quad \forall t\in [s_0,s_1).$$

{\em Recursive step.} Suppose that $z(t)\le h_k(p,t)$ for all $s_k\le t< s_{k+1}$. Since $z(t)= p+\int_{t_0}^{t} a(s) y(s) ds+ \sum_{s\in \sigma \cap (t_0,s_k]} c(s) \omega(y(s^-))$ for all $t\in [s_k,s_{k+1})$, it follows that
\begin{multline*}
  p+\int_{t_0}^{s_{k+1}} a(s) y(s) ds+ \sum_{s\in \sigma \cap (t_0,s_k]} c(s) \omega(y(s^-)) \\ =z(s_{k+1}^-)\le h_k(p,s_{k+1}).
\end{multline*}
Therefore
\begin{align*}
z(s_{k+1}) &=p+\int_{t_0}^{s_{k+1}} a(s) y(s) ds\\
  &\phantom{=p\,} + \sum_{s\in \sigma \cap (t_0,s_k]} c(s) \omega(y(s^-))+c_{k+1} \omega(y(s_{k+1}^-))\\
&=z(s_{k+1}^-)+c_{k+1}\omega(y(s_{k+1}^-))\\
&\le z(s_{k+1}^-)+c_{k+1}\omega(z(s_{k+1}^-))\\
&\le h_k(p,s_{k+1})+c_{k+1} \omega(h_k(p,s_{k+1})).
\end{align*}
Then, for all $s_{k+1}\le t<s_{k+2}$ we have that
\begin{align*}
  z(t)&=z(s_{k+1})+\int_{s_{k+1}}^{t} a(s) y(s) ds\\
      &\le z(s_{k+1})+\int_{s_{k+1}}^{t} a(s) z(s) ds \\
      &\le z(s_{k+1}) e^{\int_{s_{k+1}}^t a(s) ds}\\
      &\le h_k(p,s_{k+1}) e^{\int_{s_{k+1}}^t a(s) ds} + c_{k+1} \omega(h_k(p,s_{k+1})) e^{\int_{s_{k+1}}^t a(s) ds} \\
  &\le h_k(p,t) + c_{k+1} e^{\int_{t_0}^t a(s) ds} \left[\sup_{t_0 \le s \le t} \omega(h_k(p,s)) e^{-\int_{t_0}^s a(\tau) d\tau}\right]\\
  &= h_{k+1}(p,t)
\end{align*}
This establishes the recursive step and concludes the proof of the claim.

From the fact that $z(t)\le h_{N-1}(p,t)$ for all $s_{N-1}\le t<s_N$ and proceeding as in the recursive step it follows that also $z(t)\le h_N(p,t)$ for all $s_N\le t\le T$, which finishes the proof.

\subsection{Proof of Lemma~\ref{lem:genlem3}}
\label{sec:proof-lemma-genlem3}

The proof requires the following Claim, whose proof follows from Appendix~B of \cite{haiman_tac18} and the fact that the functions $f_{\lambda}$ and $g_{\lambda}$ satisfy items i) and ii) of Assumptions \ref{as:ass-f-1} and \ref{as:ass-g-1}, respectively.
\begin{claim}
  \label{clm:dlthbnd}
 For every $r^* > 0$ and $\eta > 0$ there exists $\kappa = \kappa(r^*,\eta) > 0$ such that for all $\lambda \in \Lambda$, $t\ge 0$, $\xi\in B_{r^*}^n$ and $\mu\in\R^m$,
  \begin{align*}
    |f_{\lambda}(t,\xi,\mu) - f_{\lambda}(t,\xi,0)| &\le \eta + \kappa \nu_f(|\mu|)\quad \text{and}\\
     |g_{\lambda}(t,\xi,\mu) - g_{\lambda}(t,\xi,0)| &\le \eta + \kappa \nu_g(|\mu|)
  \end{align*}
\end{claim}

{\bf Proof of Lemma~\ref{lem:genlem3}:}
  Fix $r>0$ and $\eta>0$, and define $r^* := \beta(r,0) \ge r$. Let $L= L(r) > 0$ be a Lipschitz constant for $f_{\lambda}(t,\cdot,0)$ on the compact set $B_{r^*}^n$ and valid for every $t\ge 0$ and every $\lambda \in \Lambda$ (such a constant exists due to iii) of Assumption \ref{as:ass-f-1}). Let $\omega=\omega_{r^*} \in \Ki$ be such that $|g_{\lambda}(t,\xi_1,0) - g_{\lambda}(t,\xi_2,0)| \le \omega(|\xi_1-\xi_2|)$ for all $\xi_1,\xi_2\in B_{r^*}^n$, all $t\ge 0$ and all $\lambda \in \Lambda$ [such a function exists due to \ref{item:0-gcont}) of Assumption~\ref{as:ass-g-1}]. Let $\kappa$ be the quantity given by Claim~\ref{clm:dlthbnd} in correspondence with $r^*$ and $\eta$. Let $x\in\T_{\Sigma_{\lambda}}(t_0,x_0,u)$ with $\lambda \in \Lambda$, $t_0 \ge 0$, $x_0\in\R^n$ and $u \in \U$ satisfy $|x(t)| \le r$ for all $t\ge t_0$. Let $y\in\T_{\Sigma_{\lambda}}(t_0,x_0,\mathbf{0})$. Then, $x(t),y(t) \in B_{r^*}^n$ for all $t\ge t_0$. Let $t \ge t_0$. For all $t_0 \le \tau \le t$, we have, using (\ref{eq:solintform}),
  \begin{multline*}
    |x(\tau) - y(\tau)| \le \int_{t_0}^\tau \Big|f_\lambda\big(s,x(s),u(s)\big) - f_\lambda\big(s,y(s),0\big)\Big| ds\\  +\sum_{s\in \gamma_{\lambda}\cap(t_0,\tau]} \Big|g_{\lambda}\big(s,x(s^-),u(s)\big) - g_{\lambda}\big(s,y(s^-),0\big)\Big|
  \end{multline*}
  Adding and subtracting $f_{\lambda}(s,x(s),0)$ and $g_{\lambda}(s,x(s^-),0)$ within the respective norm signs, employing the bound on $f_{\lambda}$ and $g_{\lambda}$ given by Claim~\ref{clm:dlthbnd} and recalling the definition of $L$ and $\kappa$, it follows that
  \begin{align*}
    |f_{\lambda}(s,x(s),u(s)) &- f_{\lambda}(s,y(s),0)|\\
    &\le \eta + \kappa \nu_f(|u(s)|)+ L|x(s) - y(s)|,\\
    |g(s,x(s^-),u(s)) &- g(s,y(s^-),0)|\\
    &\le \eta + \kappa \nu_g(|u(s)|)+ \omega(|x(s^-) - y(s^-)|).
  \end{align*}
  Defining $z(t) = |x(t) - y(t)|$, then for all $t_0 \le \tau \le t$,
  \begin{align*}
    z(\tau) &\le \int_{t_0}^t [\eta + \kappa \chi_f(|u(s)|)] ds + \sum_{s\in \gamma_{\lambda}\cap (t_0,t]} [\eta + \kappa \chi_g(|u(s)|)]\\
    &\phantom{\le}+ \int_{t_0}^\tau L z(s) ds + \sum_{s\in \gamma_{\lambda}\cap (t_0,\tau]} \omega(z(s^-))\\
    &\le \left[t-t_0+n^{\gamma_{\lambda}}_{(t_0,t]}\right]\eta + \kappa \|u_{(t_0,t]}\|_{\chi_f,\chi_g,\gamma_{\lambda}}+ L \int_{t_0}^\tau  z(s) ds\\
    &\phantom{\le}+ \sum_{s\in \gamma_{\lambda}\cap (t_0,\tau]} \omega(z(s^-))
  \end{align*}
  The result then follows from application of Lemma~\ref{lem:ggi2} (recall Remark~\ref{rem:hj0}) and the fact that  $|x(t)| \le |y(t)| + z(t) \le \beta(|x_0|,t-t_0+n^{\gamma_{\lambda}}_{(t_0,t]}) + z(t)$.

\subsection{Proof of Lemma~\ref{lem:0UBEBS}}
\label{sec:proof-lemma-0UBEBS}

  Let $\alpha$, $\rho_1$, $\rho_2$ and $c$ be as in the estimate~(\ref{eq:cubebs}). Let $\tilde\rho_1 := \max\{\rho_1,\nu_f\}$ and $\tilde\rho_2 := \max\{\rho_2,\nu_g\}$. For $r\ge 0$ define
  \begin{multline*}
    \bar\alpha(r) := \sup\Big\{ |x(t)|: x\in \T_{\Sigma_{\lambda}}(t_0,x_0,u),\\
    \lambda\in \Lambda,\; t\ge t_0\ge 0,\; |x_0|\le r,\;  \|u\|_{\lambda}\le r\Big\}
  \end{multline*}
where $\|u\|_{\lambda}:=\|u\|_{\tilde\rho_1,\tilde\rho_2,\gamma_{\lambda}}$. From this definition, it follows that $\bar\alpha$ is nondecreasing and from (\ref{eq:cubebs}) that it is finite for all $r\ge 0$.
Next, we show that $\lim_{r\to 0^+}\bar\alpha(r)=0$. Let $\beta\in \KL$ be the function which characterizes the strong 0-GUAS property of the family of systems. Let $r^*=\alpha^{-1}(2+c)$ and let $L = L(r^*) > 0$ and $\omega=\omega_{r^*}\in \Ki$ be given by Lemma~\ref{lem:genlem3} and let $h_j^0$, $j=0,1,\ldots$, be the functions defined in Lemma \ref{lem:ggi2} in correspondence with $a(s) \equiv L$ (recall Remark~\ref{rem:hj0}) and $c_j \equiv 1$. Let $\varepsilon>0$ be arbitrary. Pick $0<\delta_1<1$ such that $\delta_1 \le \beta(\delta_1,0)<\varepsilon/2$, and $T>0$ such that $\beta(\delta_1,T)<\delta_1/2$. Let $\tilde k= \lceil T \rceil+1$, where $\lceil s\rceil$ denotes the least integer not less than $s\in \R$. Since $h_{\tilde k}^0$ is continuous and $h_{\tilde k}^0(0,t)=0$ for all $t\ge 0$, then there exists $\tilde \delta>0$ such that $h_{\tilde k}^0(\tilde \delta,\tilde k)< \delta_1/2$. Define $\eta= \frac{\tilde \delta}{2 \tilde k}$ and let $\kappa = \kappa(r^*,\eta) > 0$ be given by Lemma~\ref{lem:genlem3}. Set $\delta_2 = \min\{\frac{\tilde \delta}{2 \kappa},1\}$.

Then, for every $x\in \T_{\Sigma_{\lambda}}(t_0,x_0,u)$, with $\lambda \in \Lambda$, $t_0\ge 0$, $|x_0|\le \delta_1$, $\|u\|_{\lambda}\le \delta_2$, we claim that $|x(t)|<\varepsilon$ for all $t\ge t_0$. First, note that under the given bounds for $x_0$ and $u$, from (\ref{eq:cubebs}) it follows that $\alpha(|x(t)|) \le \delta_1 + \delta_2 + c \le 2 + c$, and hence $|x(t)| \le r^*$ for all $t\ge t_0 \ge 0$. Consider the sequence $t_0<t_1<t_2<\cdots$, recursively defined as follows:
$$ t_{j+1}=\inf\{t\ge t_{j}:t-t_{j}+n^{\gamma_{\lambda}}_{(t_j,t]}\ge T\}, \quad j\ge 0.$$

We note that $T\le t_{j+1}-t_{j}+n^{\gamma_{\lambda}}_{(t_j,t_{j+1}]}\le \tilde k$ and that $t_j\to \infty$ \citep[see the proof of][Lemma~3.3]{haiman_rpic19}.  Let $I_j=(t_{j},t_{j+1}]$ for $j\ge 0$. 
The application of Lemma~\ref{lem:genlem3} with $\chi_f = \tilde\rho_1$ and $\chi_g = \tilde\rho_2$ gives the estimate (\ref{eq:genlem3bnd}) for all $t\ge t_0$. Then, by taking into account that $n^{\gamma_{\lambda}}_{(t_0,t]}\le \tilde k$ for all $t\in I_0$, that the functions $h_j^0$ are separately increasing in their arguments and $h_j^0 \le h_{j+1}^0$ for all $j\ge 0$, the definitions of $T$, $\eta$, $\delta_1$ and $\delta_2$, and (\ref{eq:genlem3bnd}), it follows that for all $t\in I_0$
\begin{align*}
  |x(t)| &\le \beta(|x_0|,0) + h_{\tilde k}^0(\tilde \delta,\tilde k)<\frac{\varepsilon}{2}+\frac{\delta_1}{2}\le \varepsilon\quad\text{and}\\
 |x(t_1)| &\le \beta(|x_0|,T) + h_{\tilde k}^0(\tilde \delta,\tilde k)<\frac{\delta_1}{2}+\frac{\delta_1}{2}\le \delta_1.
 \end{align*} 

By using recursively the same argument on each interval $I_j$ we obtain than $|x(t)|<\varepsilon$ for all $t\in I_j$ and $|x(t_{j+1})|<\delta_1$. In consequence, $|x(t)|<\varepsilon$ for all
$t\ge t_0$ as we claim.
Thus, if $\delta=\min\{\delta_1,\delta_2\}$, for all
$x\in \T_{\Sigma_{\lambda}}(t_0,x_0,u)$, with $\lambda\in \Lambda$, $t_0\ge 0$, $|x_0|\le \delta$ 
and $\|u\|_{\lambda}\le \delta$, we have $|x(t)|\le \varepsilon$ for all $t\ge t_0$. Therefore,
$\bar\alpha(r)\le \bar\alpha(\delta)<\varepsilon$ for all
$0<r<\delta$ and $\lim_{r\to 0^+}\bar\alpha(r)=0$.
 
Since $\bar\alpha$ is nondecreasing and
$\lim_{r\to 0^+}\bar\alpha(r)=0$ there exists $\hat \alpha \in \Ki$ such
that $\hat \alpha(r)\ge \bar\alpha(r)$ for all $r\ge 0$. Let
$x \in \T_{\Sigma_{\lambda}}(t_0,x_0,u)$ with $\lambda \in \Lambda$, $t_0\ge 0$, $x_0 \in \R^n$ and $u\in \U$. Let $t\ge t_0$ and let $u_{(t_0,t]}$ be the input which coincides with $u$ on $(t_0,t]$ and is zero elsewhere. From well-known results on differential equations, there exists
$x^*\in \T_{\Sigma_{\lambda}}(t_0,x_0,u_{(t_0,t]})$ such that $x^*(\tau)=x(\tau)$ for all
$\tau \in [t_0,t]$. By using the definition of $\bar\alpha$ and the
fact that $\hat \alpha(r)\ge \bar\alpha(r)$, we then have
  $|x(t)| = |x^*(t)| \le \hat\alpha(|x_0|) + \hat\alpha(\|u_{(t_0,t]}\|_{\lambda})$.
Define $\tilde\alpha \in \Ki$ via $\tilde\alpha(s) = \hat\alpha^{-1}(s)/2$. Applying $\tilde\alpha$ to both sides of the preceding inequality and using the fact that $\tilde\alpha(a+b) \le \tilde\alpha(2a) +\tilde\alpha(2b)$, we reach
  $\tilde\alpha(|x(t)|) \le |x_0| + \|u_{(t_0,t]}\|_{\lambda}$,
which establishes the result.

\section{Conclusions}
\label{sec:conclusions}

We have considered a strong version of asymptotic stability for time-varying impulsive systems whereby the convergence to zero of a state trajectory depends not only on elapsed time but also on the number of jumps that occur. In this setting, we have established that strong ISS implies strong iISS. This implication is established without resorting to any type of Lyapunov function because the latter may not exist for the type of systems considered. Future work may consider determining to what extent the current results may apply when stability is understood in the usual (weak) sense.



\section*{References}

\bibliography{/home/hhaimo/latex/strings.bib,/home/hhaimo/latex/complete_v2.bib,/home/hhaimo/latex/Publications/hernan_v2.bib}
\end{document}